%% file: BPM-Template.tex
\documentclass[]{BioPhysMath}
\usepackage[OT4,T1]{fontenc}
\usepackage[utf8]{inputenc}
\usepackage{polski}
\usepackage{graphicx}
\usepackage{booktabs}
\usepackage{amsmath}
\usepackage{amssymb}
\usepackage{placeins} 
\usepackage{mathabx}

\usepackage[english,polish]{babel} %PS

\usepackage{color} 
\usepackage{tabularx}
\usepackage{mathrsfs}
\RequirePackage[numbers,sort&compress]{natbib}%MT

\usepackage[colorlinks=true]{hyperref}
  \hypersetup{
    pdftitle={Template BioPhysMath}, %%<--insert by BioPhysMath
    pdfauthor=Nowy Autor, %%<--insert by BioPhysMath
    colorlinks,
    urlcolor=blue,
    filecolor=magenta,
    citecolor=green, 
    linkbordercolor={1 1 1}, % set to white
    citebordercolor={1 1 1},  % set to white
    urlbordercolor={ 1 1 1}  % set to white
  } 

\newcommand{\orcid}[1]{\href{https://orcid.org/#1}{\includegraphics[scale=.05]{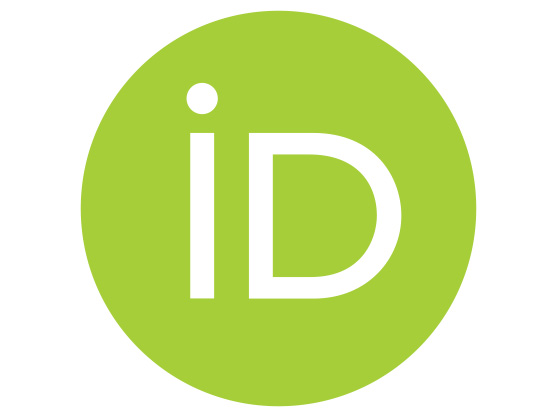}}}

\pdfoutput=1
\usepackage{graphicx}
\usepackage{wrapfig}
\graphicspath{{./Figures/},{./Pictures/}}
\usepackage{multicol}
\firstpage{i}    %inserted by BioPhysMath
\LogoG{\vspace*{1cm}\includegraphics[width=0.1\textwidth]{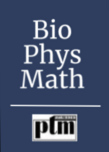}}
\volume{2}%inserted by BioPhysMath
\fasc{1}%inserted by BioPhysMath
\years{2025}%inserted by BioPhysMath
\LogoGMS{\includegraphics[width=0.18\textwidth]{Figures/bpm.png}}
\wwwtrue
\newtheorem{thm}{Theorem}
\newtheorem{cor}{Corollary}
\newtheorem{de}{Definition}

\newtheorem{prop}{Proposition}

\def\E{\mathbb{E}}

\def\P{\mathbb{P}}
\def\R{\mathbb{R}}
\def\N{\mathbb{N}}

\def\D{\mathrm{d}}
\def\Var{\mathbb{V}ar}

\def\Cov{\mathbb{C}ov}

\DeclareMathOperator*{\for}{\quad\text{for}\quad}

\newcommand\ind[1]{\mathbb{I}_{#1}}

\pages{15}
\received[25th August 2025]{9th April 2025}%inserted by BioPhysMath
\lastrevision{}
\logo{}{}{}%Do not change 9 above lines
\title[Random  Measures and ANOVA Models]{ Random  Measures, ANOVA Models\\ and Quantifying Uncertainty
in Randomized Controlled Trials}
\tpauthortrue % PS11lis12
\author[Bastian]{Caleb Deen Bastian \orcid{0000-0002-2027-7491}}%inserted by Authors
\affiliation{Princeton University, Program in  Applied and Computational Mathematics}%inserted by Authors
\address{Fine Hall       
Washington Road, Princeton, NJ  08544, USA  }%inserted by Authors
\email{cbastian@princeton.edu}%inserted by Authors
\city{Princeton, NJ}%inserted by Authors

\author[Rabitz]{Herschel Rabitz \orcid{0000-0002-4433-6142}}%inserted by Authors
\affiliation{Princeton University, Department of Chemistry}
\address{Corner of Washington Road and Scholar Way, Princeton, NJ  08544, USA  }%inserted by Authors
\email{hrabitz@princeton.edu}%inserted by Authors
\city{Princeton, NJ}%inserted by Authors

\author[Rempa{\l}a]{Grzegorz A. Rempa{\l}a\orcid{0000-0002-6307-4555}}%inserted by Authors
\affiliation{The Ohio State University, Division of Biostatistics and Department of Mathematics}
\address{1800 Cannons Dr, Columbus, OH 43210 , USA }%inserted by Authors
\email{rempala.3@osu.edul}%inserted by Authors
\city{Columbus, OH}%inserted by Authors

\comm{Ma\l{}gorzata Bogdan}%inserted by BioPhysMath 
%inserted by Authors
%inserted by Authors

\begin{document}
\vspace{-10ex}
\renewcommand{\thefootnote}{}
\footnote{\href{http://creativecommons.org/licenses/by/3.0/}{Licensed under a Creative Commons Attribution License (CC-BY)}}
%\Poczatek
%\Chapter
%\pagenumbering{roman}
\setcounter{page}{1} %%This command starts the numerations of pages
\selectlanguage{english}\Polskifalse
%\selectlanguage{polish}\Polskitrue

\begin{abstract} %inserted by Authors
This short paper introduces a novel approach to global sensitivity analysis of stochastic models, grounded in the variance-covariance structure of random variables derived from random  measures. The proposed methodology facilitates the application of information-theoretic rules for uncertainty quantification, offering several advantages. Specifically, the approach provides insights into the decomposition of variance within discrete subspaces, similar to the standard ANOVA analysis. To illustrate this point, the method is applied to datasets obtained from the analysis of randomized controlled trials  on evaluating the efficacy of the COVID-19 vaccine and assessing clinical endpoints in a lung cancer study.
\end{abstract}

%\tableofcontents
Global sensitivity analysis, as described  by Sobol's theory (see, for instance, \citealt{Sobol93}), is a powerful methodology employed in various fields, ranging from environmental modeling to complex systems analysis. Sobol's approach goes beyond traditional sensitivity analyses by systematically assessing the contributions of individual input variables and their interactions to the overall variability of model outputs. This method allows researchers to distinguish between the relative importance of different factors, aiding in the identification of critical parameters that significantly influence model behavior \citep{saltelli2008global,Saltelli10}.

At the core of Sobol's theory is the decomposition of the total variance of the model output into components attributed to individual input variables and their interactions, akin to the well-known ANOVA decompositions in the Gaussian linear models theory. By quantifying the proportion of variance associated with each factor, Sobol's sensitivity indices provide valuable insights into the sensitivity of the model to changes in specific parameters \citep{outside20,lu2018estimation}. This comprehensive understanding enables practitioners to prioritize resources, refine models, and focus on the most influential factors, ultimately improving the robustness and reliability of simulations and predictions in complex systems \citep{outside21,li12}. Sobol's variance-based global sensitivity analysis has become a cornerstone in the toolkit of modelers and analysts, offering a systematic and rigorous approach to unraveling the intricacies of complex systems. While these analyses often focus on individual input variables, the variances of random variables formed from {\em random measures} are less commonly studied, despite many potential advantages of their  usage. For a comprehensive review of the random measures theory, see for instance \cite{cinlar}. For some recent applications to complex biological systems see, for instance \cite{Bastian:2020tb,khudabukhsh2023projecting}.  

Classical sensitivity analysis is concerned with the (co)variances (uncertainties) of the effects, treating the number of data points (enrollees in the trial) as a fixed variable. The random measure setting generalizes the number of data points to be a random variable. Then, the uncertainty of the effects in random measures factors in contributions from both the individual-level of the enrollees and the randomness in the number of points. In most settings, the number of data points is fixed, which corresponds to a binomial random measure. This setting is under-dispersed, where, for example, counts in disjoint subsets of the measurement space are negatively correlated. A useful class of random measures for uncertainty quantification are orthogonal, which are de-correlated on disjoint subsets. The random measure paradigm is a natural extension and generalization of the current practice of sensitivity analysis of effects in randomized controlled trials. 

In this short note we consider  the variance of random variables formed from random  measures that admits decomposition in disjoint subsets, often referred to as `subspaces.' Global sensitivity analyses may be then  defined relative to this expansion, revealing information on the decompositions of the random measure variances in these subspaces. This approach allows for a more nuanced exploration of the impact of individual parameters on the overall variability of the model outcomes, offering insights into the specific contributions of different factors to the uncertainty in the model predictions. We show examples of application of such approach to analysis of clinical trials data. 

This article is organized as follows. First we give some minimal background (Section~2) on random measures. Then we define (Section~3) the random measure decomposition of variance and describe (Section~4) application to randomized controlled trials. We end (Section~5) with brief discussions and some conclusions. 

\section{Background} In this section we give a brief background on random counting measures, namely those possessing an orthogonal splitting property, and describe their mean and variance. 

%\subsection{Random Measure} 
Consider a random measure $N$ on the measurable space $(E,\mathscr{E})$, defined by a pair of deterministic probability measures $(\kappa,\nu)$, where $\kappa$ is a probability counting measure and $\nu$ is a measure on $(E,\mathscr{E})$. If $\nu$ is
discrete, then $N$ is referred to as a {\em random counting measure}.   Given the pair $(\kappa,\nu)$, $N$  is constructed as follows. Let $\kappa$ have mean $c$ and variance $\delta^2$.  Let $\{X_1,X_2,\dotsb\}$ be a collection of iid random variables having common law $\nu$ and independent of $\kappa$. Let counting variable $K$ follow the law $\kappa$ which is denoted as $K\sim\kappa$. The random measure $N=(\kappa,\nu)$ is formed from the random set $\mathbf{X}=\{X_1,\dotsb,X_K\}$ through \emph{stone throwing construction} \citep{cinlar,Bastian:2020tb} as \begin{equation}\label{eq:stones} N(A) =\int_EN(\D x)\ind{A}(x)= \sum_i^K\ind{A}(X_i)\for A\in\mathscr{E},\end{equation} where $\ind{A}$ is the indicator or set function. This random measure is also known in the literature as a {\em mixed binomial process} \citep{kallenberg02}

Let $\mathscr{E}_{\ge0}$ be the collection of non-negative $\mathscr{E}$-measurable functions. Take $f$ in $\mathscr{E}_{\ge0}$ and let $Nf$ be the random variable formed as \[Nf =\int_E N(\D x)f(x) = \sum_i^Kf\circ X_i.\] The mean and variance of $Nf$ are \[\E Nf = c\nu f,\quad \Var Nf = c\nu f^2 + (\delta^2-c)(\nu f)^2.\] The covariance of $f$ and $g$ in $\mathscr{E}_{\ge0}$ is given by \[\Cov(Nf,Ng) = c\nu(fg)+(\delta^2-c)\nu f\nu g.\] 
When $\kappa$ is a Poisson distribution then the random measure $N$ is referred to as a {\em Poisson random measure}. Note that in this case $\delta^2=c$ and the covariance formula simplifies. This leads to the following  key concept in the theory of  random measures.

\begin{de}[Orthogonality] Let $N$ be a random measure on $(E,\mathscr{E})$. It is said to be \emph{orthogonal} if $N(A),\dotsb,N(B)$ are uncorrelated for all choices of finite many disjoint sets $A,\dotsb,B$ in $\mathscr{E}$. 
\end{de}

Note that orthogonality conveys an orthogonal splitting property \[\Var(N(A)+\dotsb+N(B)) = \Var N(A) +\dotsb+\Var N(B)\] for disjoint  $A,\dotsb,B\text{ in }\mathscr{E}.$  When $\delta^2-c=0$, the random measure is orthogonal and \[\Cov(Nf,Ng)=c\,\nu(fg).\]

\begin{prop}[Orthogonality Condition] The random measure $N=(\kappa,\nu)$ on $(E,\mathscr{E})$ is orthogonal iff $\delta^2-c=0$.\end{prop}

\subsection{Examples of probability counting measures} For illustration we present several examples of $\kappa$ probability counting measure often utilized for  constructing a random 
measure $N$.  The key quantity of $\kappa$ is $\delta^2-c$ that  determines the correlative structure of $N$. As already noted, $N$ is orthogonal iff  $\delta^2-c=0$. In Table~\ref{tab:1}, we 
provide examples of common counting distributions  and highlight their properties. We present pairs of counting distributions exhibiting negative, zero, and positive correlations. We also 
specify whether the random measures are closed under restriction to subspaces and include some limiting relations. Note that the  orthogonal $N$ distributions  offer two choices in the table 
below, namely Poisson and Orthogonal Die (for the definition of the latter see the table caption and \citealt{bastian2020orthogonal}). It  appears the Orthogonal Die  offers a convenient alternative to the Poisson  whenever there is a theoretical need for bounded support, such as a finite bound on 
$K$ the number of points (stones) in \eqref{eq:stones}. Otherwise, the Poisson distribution is often preferred. We also note that the Orthogonal Die random measure has a Poisson limit, so in many problems with large $n$ and $m$ the distinction is stochastically negligible. 

\begin{table}[h!]
\begin{center}
\caption{Examples of counting measures. For the Orthogonal Die$(m,n)$ distribution  $\kappa$ is defined as the discrete uniform distribution on the set of consecutive integers 
$\{m,\dotsb,n\}$ for values of $m$ and $n$ satisfying mean-variance equality. For Zeta($s$) $\kappa$ is defined as  $P(\kappa=k)=\zeta(s)^{-1}k^{-s}$, where 
$\zeta(s)$ is the Riemann zeta function.}
\begin{tabular}{lclcl}
\toprule
Name & Support & $\delta^2-c$ & Closure & Limit(s)\\\midrule
Dirac$(c)$& $\{c\}$ & $-c$ & No &  \\
Binomial$(n,p)$ & $\{0,\dotsb,n\}$ & $-np^2$ & Yes & Dirac,  \\
 &  &  &  & Poisson \\
Poisson$(c)$ & $\N_{\ge0}$ & 0 & Yes &\\
Orthogonal Die$(m,n)$ & $\{m,\dotsb,n\}$ & 0 & No & Poisson \\
Negative-binomial$(r,p)$ & $\N_{\ge0}$ & $r(\frac{p}{1-p})^2$ & Yes & Poisson\\
Zeta$(s)$ for $s>2$& $\N_{>0}$ & $\frac{\zeta (s-1) \zeta (s+1)}{\zeta (s+1)^2}$  & No &\\
&  & $\frac{-\zeta (s) (\zeta (s)+\zeta (s+1))}{\zeta (s+1)^2}$  &  &\\
\bottomrule
\end{tabular}\label{tab:1}
\end{center}
\end{table}

\FloatBarrier

\subsection{Sobol systems} Consider a probability space $(E,\mathscr{E},\mu)$. Assume that $E$ is $n$-dimensional and that $\mu$ is a product probability measure. Let $f$ be a square-integrable function, i.e., in $L^2(E,\mathscr{E},\mu)$. The Sobol system representation, also known as high-dimensional model representation, of $f$, is given in terms of orthogonal component functions $\{f_u(x_u): x\subseteq\{1,\dotsb,n\}\}$ of increasing dimensions, \[f(x_1,\dotsb,x_n)=f_0+\sum_if_i(x_i)+\sum_{i<j}f_{ij}(x_i,x_j)+\dotsb+f_{1\dotsb n}(x_1,\dotsb,x_n)\] for $(x_1,\dotsb,x_n)\in E$ such that the variance decomposes as \[\Var f = \sum_i\Var f_i + \sum_{i<j}\Var f_{ij}+\dotsb+\Var f_{1\dotsb n},\] normalized as sensitivity indices \[\mathbb{S}_u = \frac{\Var f_u}{\Var f}\for u\subseteq\{1,\dotsb,n\}.\]

\section{Global sensitivity analysis} We give a brief description of analysis of variance through random measures, associated sensitivity indices, and describe some probability counting measures of interest. 

\subsection{Analysis of variance}\label{sec:var}
  
The variance of model on the full space $(E,\mathscr{E})$ is decomposed into the variance-covariance structure on a partition, obtained through the variance-covariance formulas. %Based on the class of mixed binomial processes we refer to the decomposition of random measure variance in terms of partitions the \emph{random measure analysis of variance}.

\begin{thm}[ANOVA] Let $N=(\kappa,\nu)$ be a random measure on $(E,\mathscr{E})$. Let $f\in\mathscr{E}_{\ge0}$ and consider disjoint partition $P=\{A,\dotsb,B\}$ of $E$. Then
 \begin{align}\Var Nf &= \sum_{D\in P} \Var Nf\ind{D} + \sum_{D'\ne D''\in P} \Cov(Nf\ind{D'},Nf\ind{D''})\\&=\sum_{D\in P}(c\nu(f^2\ind{D})+(\delta^2-c)(\nu(f\ind{D}))^2) + \sum_{D'\ne D''\in P} (\delta^2-c)\nu(f\ind{D'})\nu(f\ind{D''}).\nonumber\end{align} 
\end{thm}
\begin{proof} Noting that $f=f\ind{A}+\dotsb+f\ind{B}$ and that disjointedness implies that $\Cov(Nf_a,Nf_b)=(\delta^2-c)\nu f_a\nu f_b$, we have the decomposition using the variance and covariance formulas of the mixed binomial process.
\end{proof}
%\alert{GR: define binomial process}
\begin{cor}[Orthogonal ANOVA] Let $N$ be orthogonal. Then \[\Var Nf = c\nu(f^2) = c\sum_{D\in P}\nu(f^2\ind{D}).\] 
\end{cor}

\subsection{Sensitivity indices}
Normalizing by the overall variance, we retrieve the sensitivity indices which  are key in Sobol's theory of global sensitivity analysis \citep{saltelli2008global}.

\begin{samepage}
\begin{de}[Sensitivity indices of $Nf$] Conider the random measure $N$ and disjoint partition $P=\{A,\dotsb,B\}$ of $E$. The structural sensitivity index of $Nf$ relative to partition $P$ is defined as \begin{equation}\mathbb{S}_D^a\equiv\frac{\Var Nf\ind{D}}{\Var Nf}\for D\in P.\end{equation} The correlative sensitivity index is defined as \begin{equation}\mathbb{S}^b_D\equiv\sum_{F\in P: F\ne D}\frac{\Cov(Nf\ind{D},Nf\ind{F})}{\Var Nf}\for D\in P,\end{equation} where \begin{equation}1 = \sum_{D\in P}(\mathbb{S}^a_D+\mathbb{S}^b_D)= \mathbb{S}^a+\mathbb{S}^b.\end{equation}  This yields the set of sensitivity indices $\{(\mathbb{S}^a_D,\mathbb{S}^b_D): D\in P\}$. 
\end{de}
\end{samepage}

These sensitivity indices indicate the contributions of a specific partition element to the variance of the random measure in the partitions. If $c=\delta^2$ for $\kappa$, as is  the case for    Poisson \citep{Bastian:2020tb} or other orthogonal random measures, then $\mathbb{S}^a=1$ and $\mathbb{S}^b=0$, so $(\mathbb{S}^a_D)$ is a probability vector, conveying a distribution of uncertainty on the partition. Otherwise, the vector is positively or negatively defective and hence looses a probabilistic interpretation.

%When $\mathbb{S}^a=1$, then the vector $(\mathbb{S}^a_d)$ encodes the distribution of the number of points in $\{A,\dotsb,B\}$.

A common setting is the Dirac probability counting measure $\kappa=\delta_n$ for some number $n\in\N_{>0}$ in the binomial process. For the binomial process, \begin{equation}\label{eq:binsa}\mathbb{S}^a = \sum_{D\in P}\frac{\Var Nf_d}{\Var Nf}=\sum_{D\in P}\frac{\Var f_d}{\Var f}>1\end{equation} and  \begin{equation}\label{eq:binsb}\mathbb{S}^b = \sum_{D\in P}\sum_{D'\in P: D'\ne D}\frac{-\nu(f\ind{D})\nu(f\ind{D'})}{\Var f}<0.\end{equation} 
 
Unlike the example above, the orthogonal random measures produce variance-based normalized sensitivity indices in the range of [0,1]. This enables us to introduce the   concept of {\em sensitivity probability measures}.

 \begin{samepage}
\begin{de}[Sensitivity probability measure]For orthogonal random measure $N$ and function  $f$ in 
$\mathscr{E}_{\ge0}$, the sensitivity probability measure $\mathbb{S}$ on $(E,\mathscr{E})$ is defined by \begin{equation}\label{eq:sdx}\mathbb{S}(\D x) \equiv \nu(\D x)f^2(x)/\nu f^2\end{equation} such 
that \[\mathbb{S}(D)\equiv\int_D\mathbb{S}(\D x)\for D\in\mathscr{E}.\] Assume now that $E=\bigtimes_i E_i$. For variable subset $u\subseteq\{1,\dotsb\}$, denoting $E_{-u}\equiv\bigtimes_{i\in\{1,\dotsb\}: i\notin u}E_i$, 
the marginal sensitivity probability measure $\mathbb{S}_u$  is defined as \begin{equation}\label{eq:sdxi}\mathbb{S}_u(\D x_u) \equiv \int_{E_{-u}}\mathbb{S}(\D x).\end{equation} %$\mathbb{S}^a_u$ defines a projection operator $P_u$ as \[\mathbb{S}_u^a(x_u) = P_u\mathbb{S}^a(x) = \int_{E_{-u}}\D x_{-u}\mathbb{S}^a(x)\]
\end{de}
\end{samepage}

The probability measure $\mathbb{S}$ on $(E,\mathscr{E})$ reflects the distribution of uncertainty of the random variable $Nf$ on $(E,\mathscr{E})$. At times, it is beneficial to consider uncertainty in relation to a partition $P$ of the space $E$. For any such partition $P$, the  quantification of uncertainty with respect to the partition can be obtained using the entropy function
\begin{equation}\label{eq:H}
    H_P(\mathbb{S})= - \sum_{A\in P} \mathbb{S}(A) \log(\mathbb{S}(A)), \end{equation} where we
    apply the usual convention  $0\cdot\log 0 =0$. Note that $H_P$ is non-negative and bounded by $\log(|P|)$.

%Another entropy function is that of the Dirichlet distribution with positive concentration vector $\alpha=(\alpha_A: A\in P)$. We set \begin{equation}\label{eq:alpha}\alpha_A = |P|\,\mathbb{S}(A)\end{equation} such that the average concentration is unity. This is because the entropy of the Dirichlet distribution is maximized with unit concentrations. Let $n=|P|$ such that, in view of \eqref{eq:alpha}, $\alpha_0 = \sum_{A\in P}\alpha_A=|P|=n$. Then the entropy function is given by \begin{align*}H_P(\mathbb{S}) &= \log(\frac{\prod_{A\in P}\Gamma(\alpha_A)}{\Gamma(\alpha_0)}) + (\alpha_0-|P|)\psi(\alpha_0)- \sum_{A\in P}(\alpha_A-1)\psi(\alpha_A)\\&=\log(\beta(\alpha)) - \sum_{A\in P}(\alpha_A-1)\psi(\alpha_A).\end{align*} Note that $H_P$ is  non-positive and unbounded. 
%
%

\section{Randomized controlled trials}\label{ex:rct}

In the context of discussion from previous sections, we present now a framework for analysis of randomized controlled trials using random measures. In this framework the  normalized sensitivity indices  for orthogonal measures  may be interpreted as probabilities, leading to sensitivity probability measures. As we see below such measures may then provide insight into the distribution of uncertainty  across groups by using the information-theoretic uncertainty quantification.

\subsection{Random measure representation}
A randomized controlled trial (RCT) is a study where independent subjects are independently randomly assigned to different groups, e.g., control and treatment groups, and subject measurements are recorded. Thus, a randomized controlled trial is a  realization of a random  measure. 

Let $N=(\kappa,\nu)$ be a random counting measure on a clinical trial design space $(E,\mathscr{E})$, where $\kappa$ is a counting distribution with mean $c$ and variance $\delta^2$ and $\nu$ is a probability measure on $(E,\mathscr{E})$. In the simplest case, $E=\{C,T\}$ contains control and treatment groups with equal probability, i.e. $\nu\{C\}=\nu\{T\}=1/2$. The enrollees (trial participants) are  a collection of independent $E$-valued random variables $\mathbf{X}=\{X_i\}$ and their measurements  are a collection $\mathbf{Y}=\{Y_i\}$ in measurement space $(F,\mathscr{F})$ according to transition probability kernel $Q$ from $(E,\mathscr{E})$ into $(F,\mathscr{F})$, i.e. $Y_i\sim Q(X_i,\cdot)$, e.g., $(F,\mathscr{F})=(\R_{\ge0},\mathscr{B}_{\R_{\ge0}})$. The  random measure associated with the collection $(\mathbf{X},\mathbf{Y})$ is $M=(\kappa,\nu\times Q)$ defined on $(E\times F,\mathscr{E}\otimes\mathscr{F})$   by stone throwing construction as \[Mf =\int_{E\times F}M(\D x,\D y)f(x,y)= \sum_i^K f\circ (X_i,Y_i)\for f\in(\mathscr{E}\otimes\mathscr{F})_{\ge0}.\] 

Define disjoint functions $f_T(x,y)=\ind{\{T\}}(x)y$ and $f_C(x,y)=\ind{\{C\}}(x)y$. For general random counting measures $M$, the random variables $Mf_T$ and $Mf_C$ are correlated, where \[\Cov(Mf_T,Mf_C) = (\delta^2-c)(\nu\times Q)f_T(\nu\times Q)f_C,\] with mean \[\E Mf_T = c(\nu\times Q)f_T\] and variance \[\Var Mf_T = c(\nu\times Q)f_T^2 + (\delta^2-c)((\nu\times Q)f_T)^2.\] The variance of the random measure describes its uncertainty. We have mean \[(\nu\times Q)f_T = \int_E\nu(\D x)\ind{\{T\}}(x)\int_FQ(x,\D y) y=\nu\{T\}\int_F Q(\{T\},\D y)y=\nu\{T\}c_T\] and second moment
\begin{align} (\nu\times Q)f_T^2  & = \int_E\nu(\D x)\ind{\{T\}}(x)\int_FQ(x,\D y) y^2\\ \nonumber & =\nu\{T\}\int_F Q(\{T\},\D y)y^2=\nu\{T\}(c_T^2+\delta_T^2). \end{align}

The null hypothesis is that group means are the same, $(\nu\times Q)f_T=(\nu\times Q)f_C$, thus depending on the mean measure. This is equivalent to $\E Mf_T = \E Mf_C$. The typical setting is Dirac probability counting measure $\kappa=\text{Dirac}(c)$ for $c\in\N_{>0}$. The Dirac random measure has minimum variance, $\delta^2=0$. The covariance is negative due to the fixed sample size, \[\Cov(Mf_T,Mf_C) = -\frac{c}{4}c_Tc_C\] and \[\Cov(\frac{1}{c}Mf_T,\frac{1}{c}Mf_C) = -\frac{1}{4c}c_Tc_C\] decaying to zero as $c\rightarrow\infty$. 

We note that if  the random measure $N$ is orthogonal (for instance, it is a Poisson random measure with the Poisson counting distribution as listed in Table~\ref{tab:1}), then we have the following properties: (i) the covariance is zero for any mean sample size $c\in(0,\infty)$, (ii) the variance is given by \[\Var Mf_T = c(\nu\times Q)f_T^2\] and (iii) the normalized variances form a sensitivity distribution \begin{equation}\label{eq:sT}
    \mathbb{S}_T = \frac{(\nu\times Q)f_T^2}{(\nu\times Q)f^2} = \frac{\nu\{T\}(c_T^2+\delta_T^2)}{\nu\{T\}(c_T^2+\delta^2_T) + \nu\{C\}(c_C^2 + \delta_C^2)},\end{equation} where $\mathbb{S}_T+\mathbb{S}_C=1$. 
For \( \mathbb{S}=(\mathbb{S}_T,\mathbb{S}_C) \) the entropy function is $H(\mathbb{S})=-\mathbb{S}_T\log(\mathbb{S}_T)-\mathbb{S}_C\log(\mathbb{S}_C)$.  
%\alert{Entropy def here}
%We take $\nu\{T\}=\nu\{C\}=1/2$, $c_T = a c_C$ and $\delta_T^2 = b\delta_C^2$ for $a,b\in(0,\infty)$ and plot $H(\mathbb{S})$ as a function of $a$ and $b$ in Figure~\ref{fig:entcont} where $c_C=\delta_C^2=1$. 
%
%\begin{figure}[h!]
%\centering
%\includegraphics[width=4in]{entropycont.pdf}
%\caption{Entropy of $\mathbb{S}$ for a randomized controlled trial with non-negative measurements using an orthogonal random measure}\label{fig:entcont}
%\end{figure}
%
%\FloatBarrier

\subsection{Testing vaccine efficacy}

For example, consider a randomized controlled trial testing vaccine efficacy. The measurement is an indicator of infection  of the underlying disease being vaccinated against over the period of the trial. Then $c_T = \P(T)$ and $c_C=\P(C)$ are the probabilities of infection based on group assignment. The standard measure of vaccine efficacy is usually (see, e.g.,\citealt{andersson2012stochastic}) \begin{equation}\label{eq:Eff}
    \mathrm{Eff}=1-\frac{\P(T)}{\P(C)}.\end{equation}
    The corresponding empirical risk measure of uncertainty about the vaccine efficacy  may be  then defined as
\begin{equation}\label{eq:unc} \mathrm{Unc}(\mathrm{Eff}) = 2\min(1 - \mathrm{Eff}, \mathrm{Eff}), \end{equation}
which reflects the fact that, assuming $\mathrm{Eff} \in [0, 1]$, the uncertainty should be lowest near the boundaries of the interval and highest at the midpoint, $\mathrm{Eff} = 0.5$.
    We will consider an alternative to \eqref{eq:unc} based on \( \mathbb{S}\) and  \(H(\mathbb{S})\). To this end 
note that  \[\Var Mf_T = c\nu\{T\}\P(T) + (\delta^2-c)\nu^2\{T\}\P^2(T)\] and the covariance is given by \[\Cov(Mf_T,Mf_C) = (\delta^2-c)\nu\{T\}\nu\{C\}\P(T)\P(C).\] 

 For orthogonal $M$  (say for $\kappa$ Poisson distributed) the variance is \[\Var Mf_T = c\nu\{T\}\P(T)\] and the sensitivity index in $T$ is \[\mathbb{S}_T = \frac{\nu\{T\}\P(T)}{\nu\{T\}\P(T)+\nu\{C\}\P(C)}.\] When $\nu\{T\}=\nu\{C\}=1/2$ and  $\P(T)\ll\P(C)$, then the uncertainty of the effects is dominated by the infections of the control group, $\mathbb{S}_T\ll\mathbb{S}_C$.  In this case $ \mathbb{S}_C\approx \mathrm{Eff}$ where $\mathrm{Eff}$ is given by \eqref{eq:Eff}. 

Clinical trials of highly efficacious vaccines against infectious diseases (i.e., those with high vaccine efficacy \( \mathrm{Eff} \)) are expected to exhibit low uncertainty, as measured by the entropy \( H(\mathbb{S}) \). This appears to be the case in the current study. To enable a more direct comparison with the uncertainty measure defined in equation~\eqref{eq:unc}, it is helpful to consider a scaled version of the entropy:
\begin{equation}\label{eq:H2}
H_2(\mathbb{S}) = \frac{H(\mathbb{S})}{\log(2)},
\end{equation}
which corresponds to computing the entropy using logarithms of base 2. This normalization aligns with the binary entropy function and allows the uncertainty to be interpreted in units of bits.

As our  example we consider the recent Moderna and Pfizer vaccines for COVID-19 \cite{moderna,pfizer}. For Moderna, there are approximately $n=30\,400$ enrollees, where infection was recorded in 5 cases in the vaccine group and 90 cases in the control group, giving $\mathrm{Eff}=17/18\approx 0.944$,  the sensitivity indices \[\mathbb{S}_C = \frac{18}{19}\approx 0.947,\quad \mathbb{S}_T=\frac{1}{19}\approx 0.053\] and the binary  entropy \[H_2(\mathbb{S})=\frac{18}{19}\log_2(19/18) + \frac{1}{19}\log_2(19)\approx 0.298.\] 
The approximate 95\% confidence intervals  for $\mathbb{S}_C$ and $H_2(\mathbb{S})$ are,  respectively, $(0.945,  0.950)$ and  $(0.287,  0.307)$.
%$(0.199 ,  0.213)$. 
They may be obtained  by simulating multiple times (here 10,000 times) from the  mixed binomial process with $Poisson(n)$ trials  and the probability of success $\mathbb{S}_C$. A slightly wider and  more  statistically appropriate (but also more complicated) approximation to this interval could be obtained by directly simulating from the random measure (using $\P(T)$ and $\P(C)$) and not the sensitivity measure (using $\mathbb{S}_C$). However, the difference appears small and largely irrelevant for the current discussion. 

For Pfizer, there are approximately $n=44\,000$ enrollees, with 8 cases in the vaccine group and 162 cases in the control group, giving $\mathrm{Eff}={77}/{81}\approx 0.951$, the  sensitivity indices \[\mathbb{S}_C = \frac{81}{85}\approx 0.953,\quad \mathbb{S}_T=\frac{4}{85}\approx 0.047\] and the  entropy \[H_2(\mathbb{S})=\frac{81}{85}\log_2(85/81) + \frac{4}{85}\log_2(85/4)\approx0.274.\]   
The approximate 95\% confidence intervals for $\mathbb{S}_C$ and $H(\mathbb{S})$ are, respectively, $(0.951 ,  0.955 )$ and $(0.265, 0.283)$
%(0.184,  0.196) 
obtained as above.
The values of the binary entropies for Moderna and Pfizer trials as measures of uncertainty of, respectively, 0.298 and 0.274  may be now compared with the corresponding risk-based uncertainty values   Unc(17/18)=0.111  and Unc(77/81)=0.099  based on $\mathrm{Eff} $ of and obtained from  \eqref{eq:unc}.  The more detailed comparison of  $\mathrm{Unc}(p)$  and  $H_2(p)$ functions is provided in Figure~\ref{fig:example}.  

\begin{figure}[htbp] %  figure placement: here, top, bottom, or page
   \centering
   \includegraphics[width=.7\linewidth]{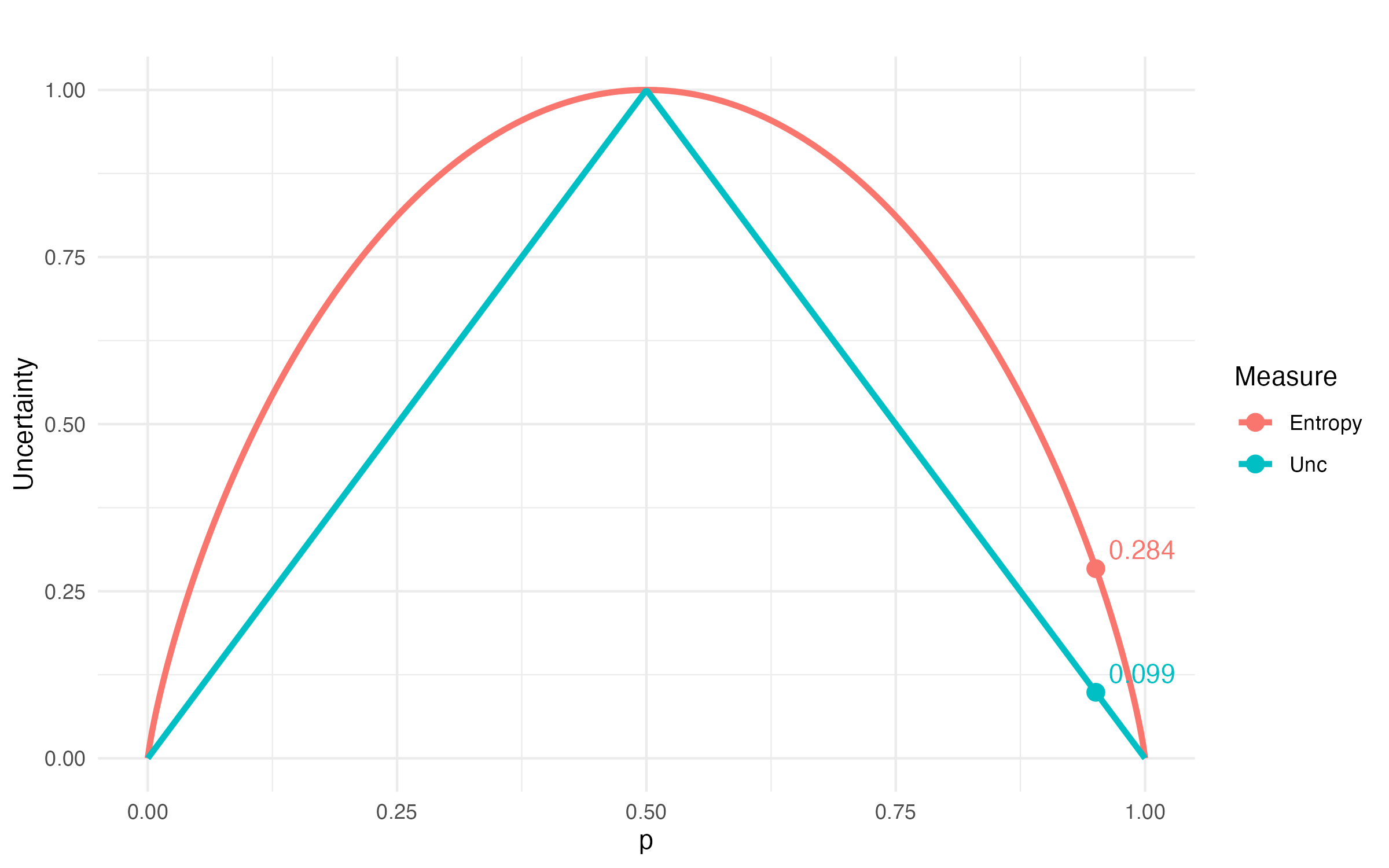} 
   \caption{The comparison of the uncertainty measures \( \mathrm{Unc}(\cdot) \) and the binary entropy \( H_2(\cdot) \) illustrates their differing behavior. For the purpose of comparison, we assume \( \mathrm{Eff} = \mathbb{S}_C = p \). As see from the graph, the binary entropy \( H_2(\cdot) \) is generally regarded as a more conservative measure of uncertainty as compared to \( \mathrm{Unc}(\cdot)\). Approximate numerical values for the Pfizer trial (see Table~\ref{tab:2}), corresponding to \( p = \frac{77}{81} \), are plotted for reference.}
   \label{fig:example}
\end{figure}

\subsection{Clinical end-points in a cancer study}
Recall  the sensitivity index provided in \eqref{eq:sT}. In this section, we will illustrate how the entropy-based analysis, conducted in the previous section, can be also employed to cluster qualitatively different information found in various types of endpoints within a clinical trial. This question becomes particularly relevant when comparing the multiple effects of the treatment with those of a control group.   This will be illustrated with a specific example. 

% \[\mathbb{S}_T = \frac{(\nu\times Q)f_T^2}{(\nu\times Q)f^2} = \frac{\nu\{T\}(c_T^2+\delta_T^2)}{\nu\{T\}(c_T^2+\delta^2_T) + \nu\{C\}(c_C^2 + \delta_C^2)}\] where $\mathbb{S}_T+\mathbb{S}_C=1$.
 
 In the clinical trial `A Study in Non-Small Cell Lung Cancer' by Eli Lilly and Company (identifier NCT01232452), various clinical end-points are measured \citep{lungcancer}. These are shown below in Table~\ref{tab:2}, with either 95\% confidence intervals or standard deviation, in months otherwise indicated, displayed along with the sensitivity indices, entropies, and uncertainties to three decimal places. For example, consider `progression-free survival' with treatment group averaging 5.22 months with 95\% confidence interval 4.24 to 6.74 months and the control group averaging 5.45 months with 95\% CI 3.88 to 6.05 months. Setting $\nu\{T\}=\nu\{C\}=1/2$, $c_T=5.22$, $\delta_T^2=(6.74-4.24)/4\approx0.637$, $c_C=5.45$, and $\delta_C^2=(6.05-3.88)/4\approx 0.553$, we obtain 
 \begin{equation}\label{eq:ST}  \mathbb{S}_T=\frac{\nu\{T\}(c_T^2+\delta_T^2)}{\nu\{T\}(c_T^2+\delta^2_T) + \nu\{C\}(c_C^2 + \delta_C^2)}\approx 0.48.\end{equation} Most of end-points are high entropy effects, that is, they generally have high uncertainty. The entropies cluster with three end-points near maximal entropy (progression-free survival, time to progressive disease, and change in tumor size), two end-points with less entropy (percentage of participants achieving an objective response and duration of response), and one end-point with much smaller entropy (time to worsening of symptoms). This suggests that response is moderately uncertain, accompanied with low uncertainty of time to worsening of symptoms, while the other endpoints are highly uncertain.   
% 
%  \begin{table}[h!]
% \begin{center}
% \begin{tabular}{p{2in}|cc|cc|c}
% \toprule
% End-point & Treatment & Control & $\mathbb{S}_T$ & $\mathbb{S}_C$ &$H_2(\mathbb{S})$ \\\midrule
% Progression-free survival & 5.45 (3.88 to 6.05) & 5.22 (4.24 to 6.74)& 0.520 & 0.480 & 0.998 \\
% Percentage of participants achieving an objective response & 37.9 (27.7 to 49.0)&30.6 (21.0 to 41.5) & 0.603 & 0.397 & 0.969 \\
% Duration of response & 4.90 (4.17 to 6.28) & 3.91 (2.92 to 6.41) & 0.602 & 0.398 & 0.969 \\
% Time to progressive disease & 6.05 (5.32 to 7.79) & 6.05 (4.93 to 7.89) & 0.499 & 0.501 & 0.999 \\
% Time to worsening of symptoms & 2.14 (1.54 to 2.99) & 4.21 (2.43 to 5.36) & 0.205 & 0.795 & 0.731 \\
% Change in tumor size & -23.88 (18.859) & -16.04 (26.143) & 0.489 & 0.511 & 0.999 \\
% \bottomrule
% \end{tabular}\caption{Treatment and control measured clinical end-points and computed sensitivity indices and entropies for the randomized controlled trial `A Study in Non-Small Cell Lung Lancer' (NCT01232452) \cite{lungcancer}. All time variables are reported in units of one month. The values of $\mathbb{S}_T$ and $H_2(\mathbb{S})$ are computed according to the formulas \eqref{eq:ST} and \eqref{eq:H2}, respectively.}\label{tab:2}
% \end{center}
% \end{table}

\begin{table}[h!]
\caption{Treatment and control measured clinical end-points and computed sensitivity indices and entropies for the randomized controlled trial \textit{A Study in Non-Small Cell Lung Cancer} (NCT01232452) \cite{lungcancer}. All time variables are reported in months. The values of $\mathbb{S}_T$ and $H_2(\mathbb{S})$ are computed according to \eqref{eq:ST} and \eqref{eq:H2}.}
\centering
\begin{tabular}{p{0.85in}|c|c|c|c|c}
\toprule
End-point & Treatment & Control & $\mathbb{S}_T$ & $\mathbb{S}_C$ & $H_2(\mathbb{S})$ \\
\midrule
Progression-free survival & 5.45 (3.88--6.05) & 5.22 (4.24--6.74) & 0.520 & 0.480 & 0.998 \\
Objective response (\%) & 37.9 (27.7--49.0) & 30.6 (21.0--41.5) & 0.603 & 0.397 & 0.969 \\
Duration of response & 4.90 (4.17--6.28) & 3.91 (2.92--6.41) & 0.602 & 0.398 & 0.969 \\
Time to progressive disease & 6.05 (5.32--7.79) & 6.05 (4.93--7.89) & 0.499 & 0.501 & 0.999 \\
Time to worsening symptoms & 2.14 (1.54--2.99) & 4.21 (2.43--5.36) & 0.205 & 0.795 & 0.731 \\
Change in tumor size & -23.88 (18.9) & -16.04 (26.1) & 0.489 & 0.511 & 0.999 \\
\bottomrule
\end{tabular}

\label{tab:2}
\end{table}

\FloatBarrier

\section{Discussion and conclusions} Random measures, as stochastic objects, have well-defined means and variances whenever they exist. The variance of a random measure encodes its uncertainty. A particularly useful subset of random measures is that of orthogonal random measures, which exhibit an orthogonal splitting property. This property allows for the definition of sensitivity probability measures, which encode the uncertainty associated with the action of the random measure on the risk function in the respective coordinates. Once such measures are available, their entropy and normalized entropy can be readily determined and used to quantify uncertainty and cluster various stochastic outcomes.

In this paper, we illustrate with examples that specific clinical trials, including those for vaccines, can be understood as instances of random counting measures. We also introduce random measures within the context of uncertainty related to observable clinical endpoints, especially pertinent in trials such as those centered around cancer treatment. We calculate the mean and covariance of the resulting random variables, derived from the integral action of the random measure on the function. We argue that in the realm of orthogonal random measures, sensitivity probability measures capture the distribution of uncertainty (variance) in the coordinates, incorporating, at a minimum, the group labels associated with clinical trials.

In the examples explored in this paper, we find that entropy-based analysis of sensitivity measures effectively captures uncertainty across clinical trials. In COVID-19 vaccine efficacy studies, minimal uncertainty is observed in the effects of the (Boolean) indicator functions for the groups. In contrast, a cancer clinical trial reveals substantial variability in the effects observed in measured clinical endpoints, aligning with the inherent challenge of extracting signals from data characterized by significant variability. Furthermore, the examples illustrate how the uncertainty measures we propose can support clinical decision-making by quantifying the robustness of outcomes across multiple endpoints. This capability is particularly valuable in complex clinical contexts, where decisions often hinge on a combination of clinical signals rather than a single metric. By offering a principled way to evaluate the stability of results across such outcomes, our framework may complement existing evaluation tools and enhance interpretability in high-stakes settings. In the context of COVID-19 vaccine trials, we also provide a comparison with a  classical risk-based uncertainty assessment method (Figure 1), further highlighting the distinct insights afforded by our entropy-based approach.

In summary, our investigation highlights the utility of random measures, particularly orthogonal random measures, in elucidating the uncertainty inherent in clinical trials. By introducing sensitivity probability measures and employing entropy-based analyses, we provide a comprehensive framework for understanding and quantifying uncertainty across diverse trial scenarios. The examples presented, ranging from COVID-19 vaccine efficacy studies to cancer clinical trials, underscore the adaptability of our approach to different medical contexts. As we continue to refine these methodologies, our findings indicate that mathematical theory of random measures may  contribute to the ongoing efforts to enhance the precision and reliability of clinical trial assessments, fostering advancements in evidence-based medicine.

%%% AUTHOR: optional acknowledgments here
\vspace{2ex}
\acknowledgments{GAR work was patially supported by the HEALMOD initiative at The Ohio State University}
%\authorcontributions{All authors equally contributed to the paper.}
%%%%%%%%%%%%%%%%%%%%%%%%%%%%%%%%%%%%%%%%%%
%\funding{This research received no external funding. }
%%%%%%%%%%%%%%%%%%%%%%%%%%%%%%%%%%%%%%%%%%
%\conflictsofinterest{The author declare no conflict of interest.} 
\vskip 0.5cm

\input{yourBibTeXfile.bib}
\centerline{\bf\large References}
%\small
\bibliographystyle{abbrvnat}
\bibliography{\jobname}
%\bibliography{refs.bib}
\bigskip%\normalsize

\Koniec
\end{document}

%% file: yourBibTeXfile.bib
% Please copy or paste your bibtex file here 
%If possible add DOI NUMBER to each item

\begin{filecontents}[overwrite]{\jobname.bib}

@book{saltelli2008global,
	author = {Saltelli, A. and Ratto, M. and Andres, T. and Campolongo, F. and Cariboni, J. and Gatelli, D. and Saisana, M. and Tarantola, S.},
	date-added = {2024-01-10 14:13:55 -0600},
	date-modified = {2024-01-10 14:15:35 -0600},
	publisher = {{John Wiley \& Sons}},
	title = {Global sensitivity analysis: the primer},
	year = {2008},
	doi = {10.1002/9780470725184}
}
@book{andersson2012stochastic,
	author = {Andersson, H. and Britton, T.},
	date-added = {2024-01-10 14:08:18 -0600},
	date-modified = {2024-01-10 14:16:46 -0600},
	isbn = {9781461211587},
	publisher = {Springer New York},
	series = {Lecture Notes in Statistics},
	title = {Stochastic Epidemic Models and Their Statistical Analysis},
	year = {2012},
	bdsk-url-1 = {https://books.google.com/books?id=hTX2BwAAQBAJ},
	doi = {10.1007/978-1-4612-1158-7}
}
@techreport{lungcancer,
	author = {Eli Lilly and Company},
	date-added = {2023-11-04 10:14:28 -0500},
	date-modified = {2023-11-04 10:16:11 -0500},
	institution = {National Library of Medicine},
	number = {NCT01232452},
	title = {A Study in Non-Small Cell Lung Cancer},
	year = {2019}}

@article{Bastian:2020tb,
	annote = {https://doi.org/10.1002/mma.6224},
	author = {Bastian, Caleb Deen and Rempala, Grzegorz A.},
	booktitle = {Mathematical Methods in the Applied Sciences},
	da = {2020/05/15},
	date = {2020/05/15},
	date-added = {2022-06-01 13:44:06 -0500},
	date-modified = {2022-06-01 13:45:10 -0500},
	doi = {10.1002/mma.6224},
	isbn = {0170-4214},
	journal = {Mathematical Methods in the Applied Sciences},
	journal1 = {Mathematical Methods in the Applied Sciences},
	journal2 = {Math Meth Appl Sci},
	keywords = {Laplace functional; Poisson-type (PT) distributions; random counting measure; stone throwing construction; strong invariance; thinning},
	m3 = {https://doi.org/10.1002/mma.6224},
	month = {2022/06/01},
	n2 = {We show that in a broad class of random counting measures, one may identify only three that are rescaled versions of themselves when restricted to a subspace. These are Poisson, binomial, and negative binomial random measures. We provide some simple examples of possible applications of such measures.},
	number = {7},
	pages = {4658--4668},
	publisher = {John Wiley \& Sons, Ltd},
	title = {{Throwing stones and collecting bones: Looking for Poisson-like random measures}},
	ty = {JOUR},
	volume = {43},
	year = {2020},
	year1 = {2020}
	}

@article{outside20,
	author = {T. Ziehn and K. J. Hughes and J. F. Griffiths and R. Porter and A. S. Tomlin},
	date-added = {2021-03-21 16:02:11 -0400},
	date-modified = {2021-03-21 16:02:11 -0400},
	doi = {10.1080/13647830902878398},
	eprint = {http://www.tandfonline.com/doi/pdf/10.1080/13647830902878398},
	journal = {Combustion Theory and Modelling},
	number = {4},
	pages = {589-605},
	title = {A global sensitivity study of cyclohexane oxidation under low temperature fuel-rich conditions using {HDMR} methods},
	volume = {13},
	year = {2009}
}

@article{outside21,
	author = {Ziehn, T. and Tomlin, A. S.},
	date-added = {2021-03-21 16:02:11 -0400},
	date-modified = {2021-03-21 16:02:11 -0400},
	doi = {10.1002/kin.20367},
	issn = {1097-4601},
	journal = {International Journal of Chemical Kinetics},
	number = {11},
	pages = {742--753},
	publisher = {Wiley Subscription Services, Inc., A Wiley Company},
	title = {A global sensitivity study of sulfur chemistry in a premixed methane flame model using {HDMR}},
	volume = {40},
	year = {2008}
	}

@article{li12,
	author = {Li, Genyuan and Rabitz, Herschel},
	date-added = {2021-03-21 16:02:11 -0400},
	date-modified = {2021-03-21 16:02:11 -0400},
	doi = {10.1007/s10910-011-9898-0},
	issn = {0259-9791},
	journal = {Journal of Mathematical Chemistry},
	keywords = {HDMR; Global sensitivity analysis; D-MORPH regression; Extended bases; Least-squares regression; Orthonormal polynomial},
	language = {English},
	number = {1},
	pages = {99-130},
	publisher = {Springer Netherlands},
	title = {General formulation of HDMR component functions with independent and correlated variables},
	volume = {50},
	year = {2012}
}

@article{Saltelli10,
	author = {Saltelli, Andrea and Annoni, Paola and Azzini, Ivano and Campolongo, Francesca and Ratto, Marco and Tarantola, Stefano},
	da = {2010/02/01/},
	date-added = {2021-03-21 16:02:11 -0400},
	date-modified = {2021-03-21 16:02:11 -0400},
	doi = {10.1016/j.cpc.2009.09.018},
	isbn = {0010-4655},
	journal = {Computer Physics Communications},
	number = {2},
	pages = {259--270},
	title = {Variance based sensitivity analysis of model output. {D}esign and estimator for the total sensitivity index},
	ty = {JOUR},
	volume = {181},
	year = {2010}
}

@techreport{moderna,
	author = {FDA},
	date-added = {2020-12-18 17:53:34 -0500},
	date-modified = {2020-12-18 17:54:57 -0500},
	institution = {Moderna},
	title = {{E}mergency {U}se {A}uthorization ({EUA}) for {M}oderna {COVID}-19 {V}accine},
	year = {2020}}

@techreport{pfizer,
	author = {FDA},
	date-added = {2020-12-18 17:50:28 -0500},
	date-modified = {2020-12-18 17:52:30 -0500},
	institution = {Pfizer},
	title = {{E}mergency {U}se {A}uthorization ({EUA}) for an {U}napproved {P}roduct {R}eview {M}emorandum: {P}fizer-{B}io{NT}ech {COVID}-19 {V}accine/ {BNT}162b2},
	year = {2020}
}

@Article{Sobol93,
  author        = {Sobol, Ilya M.},
  file          = {Sobol_1990aa.pdf},
  journal       = {Matematicheskoe Modelirovanie},
  keywords      = {uncertainty,sensitivity,sobol},
  langid        = {russian},
  note          = {Mathematical Modeling \& Computational Experiment (Engl.), 1993, 1, 407–414.},
  number        = {1},
  pages         = {112–-118},
  title         = {On sensitivity estimation for nonlinear mathematical models},
  volume        = {2},
  year          = {1990}
}

@article{bastian2020orthogonal,
	author = {Bastian, Caleb Deen and {Rempala}, Grzegorz A and Rabitz, Herschel},
	date-added = {2020-10-07 21:52:29 -0500},
	date-modified = {2024-05-20 08:42:52 -0500},
	journal = {The Annals of Applied Probability},
	keywords = {Mathematics - Probability, 60G57, 60G55, 60F05},
	title = {{Random measures and orthogonal dice}},
	year = 2024,
	doi = {10.48550/arXiv.2009.10503}
}

@article{lu2018estimation,
author = {Rong Lu and Danxin Wang and Min Wang and Grzegorz A. Rempala},
title = {Estimation of {S}obol's sensitivity indices under generalized linear models},
journal = {Communications in Statistics - Theory and Methods},
volume = {47},
number = {21},
pages = {5163--5195},
year = {2018},
publisher = {Taylor \& Francis},
doi = {10.1080/03610926.2017.1388397}

}
@article{khudabukhsh2023projecting,
title = {Projecting {COVID-19} cases and hospital burden in {O}hio},
journal = {Journal of Theoretical Biology},
volume = {561},
pages = {111404},
year = {2023},
issn = {0022-5193},
doi = {10.1016/j.jtbi.2022.111404},
author = {Wasiur R. KhudaBukhsh and Caleb Deen Bastian and Matthew Wascher and Colin Klaus and Saumya Yashmohini Sahai and Mark H. Weir and Eben Kenah and Elisabeth Root and Joseph H. Tien and Grzegorz A. Rempała},
keywords = {COVID-19, Dynamical Survival Analysis, SIR model, Prediction}
}
@book{kallenberg02,
  added-at = {2015-02-23T20:19:17.000+0100},
  author = {Kallenberg, Olav},
   doi = {10.1007/978-1-4757-4015-8},
  edition = {Second},
  interhash = {ffb70ba7323cf9bae5dd8b094824820d},
  intrahash = {905a4d28c5d86d23aaf21f62d0f34954},
  isbn = {0-387-95313-2},
  keywords = {book probability_theory reference},
  mrclass = {60-01},
  mrnumber = {1876169 (2002m:60002)},
  mrreviewer = {Klaus D. Schmidt},
  pages = {xx+638},
  publisher = {Springer-Verlag, New York},
  series = {Probability and its Applications (New York)},
  timestamp = {2015-02-23T20:19:17.000+0100},
  title = {Foundations of modern probability},
    year = 2002
}

@book{cinlar,
	author = {Erhan Cinlar},
	date-added = {2018-07-25 12:49:01 +0000},
	date-modified = {2018-07-25 12:49:01 +0000},
	publisher = {Springer-Verlag New York},
	title = {Probability and Stochastics},
	year = {2011},
	doi = {10.1007/978-0-387-87859-1}
}
\end{filecontents} 